\documentclass[12pt]{amsart}
\usepackage{amscd,amsxtra,amssymb, bbm}

\newtheorem{thm}{Theorem}[section]
 \newtheorem{example}[thm]{Example}
 \newtheorem{prop}[thm]{Proposition}
 \newtheorem{defn}[thm]{Definition}
 \newtheorem{lem}[thm]{Lemma}
 \newtheorem{cor}[thm]{Corollary}
 \newtheorem{fact}[thm]{Fact}
  \newtheorem{rem}[thm]{Remark}

 \usepackage{dsfont}
 \usepackage[all]{xy}

\begin{document}

\title{Remarks on rational solutions of Yang-Baxter equations }

\author{Thilo Henrich}

\address{
Mathematisches Institut,
Universit\"at Bonn,
Endenicher Allee 60, 53115 Bonn,
Germany
}
\email{henrich@math.uni-bonn.de}
\thanks{This work was supported by the DFG project Bu-1866/2-1.}

\begin{abstract}
In this article, we study unitary rational solutions of the associative Yang-Baxter equation with three spectral parameters. We explain how such solutions arise from the geometry of vector bundles on a cuspidal cubic curve. Moreover, we investigate how these solutions are related to the quantum and classical Yang-Baxter equations.
\end{abstract}

\maketitle



\section{Introduction}

\noindent In this article we study solutions of the associative Yang-Baxter
equation (AYBE)\begin{equation}
\begin{array}{c}
r^{12}(u;y_{1},y_{2})\, r^{23}(u+v;y_{2},y_{3})=\\
=r^{13}(u+v;y_{1},y_{3})\, r^{12}(-v;y_{1},y_{2})+r^{23}(v;y_{2},y_{3})\, r^{13}(u;y_{1},y_{3}).\end{array}\label{eq:AYBE in 3 variables}\end{equation}
 Here $r:(\mathbb{C}^{3},0)\rightarrow A\otimes A$ is the germ of
a meromorphic function for $A=\mbox{Mat}_{n\times n}(\mathbb{C})$.
Moreover, for $i\neq j\in\{1,2,3\}$, we use the notation $r^{ij}=r\circ\rho_{ij}$
for the composition of $r$ with the canonical embedding $\rho_{ij}:A^{\otimes2}\rightarrow A^{\otimes3}$,
e.g. $\rho_{13}(a\otimes b)=a\otimes\mathds{1}\otimes b$. A solution
$r$ of (\ref{eq:AYBE in 3 variables}) is called unitary if \[
\begin{array}{c}
r^{12}(v;y_{1},y_{2})=-r^{21}(-v;y_{2},y_{1})\end{array}\]
and non-degenerate if the tensor $r(v;y_{1},y_{2})\in A\otimes A$
is non-degenerate for generic $v,y_{1},y_{2}$. Non-degenerate unitary
solutions of (\ref{eq:AYBE in 3 variables}) have previously been
studied by Polishchuk \cite{Polischuck2002,Polischuckb} and Burban,
Kreu\ss{}ler \cite{Burban}.

We will focus on solutions $r$ of (\ref{eq:AYBE in 3 variables})
satisfying the following Ansatz: \begin{equation}
\begin{array}{c}
r(v;y_{1},y_{2})={\displaystyle \frac{\mathds{1}\otimes\mathds{1}}{v}}+r_{0}(y_{1},y_{2})+v\, r_{1}(y_{1},y_{2})+v^{2}\, r_{2}(y_{1},y_{2})+\dots\end{array}\label{eq:Laurent assumption}\end{equation}
This is motivated by the following fact. Let $\mbox{pr}:A\rightarrow\mathfrak{sl}_{n}(\mathbb{C})$
denote the canonical projection $X\mapsto X-\frac{\mbox{tr}X}{n}\mathds{1}$.
It is not difficult to show, see \cite{Polischuck2002,Burban}, that
$ $$\bar{r}_{0}(y_{1},y_{2})=(\mbox{pr}\otimes\mbox{pr})\left(r_{0}(y_{1},y_{2})\right)$
satisfies the classical Yang-Baxter equation (CYBE)\begin{equation}
\begin{array}{c}
{\displaystyle \frac{{}}{{}}}\Bigl[\bar{r}_{0}^{12}(y_{1},y_{2}),\bar{r}_{0}^{23}(y_{2},y_{3})\Bigr]+\Bigl[\bar{r}_{0}^{12}(y_{1},y_{2}),\bar{r}_{0}^{13}(y_{1},y_{3})\Bigr]=\\
{\displaystyle \frac{{}}{{}}}=-\Bigl[\bar{r}_{0}^{13}(y_{1},y_{3}),\bar{r}_{0}^{23}(y_{2},y_{3})\Bigr].\end{array}\label{eq:CYBE}\end{equation}
In section \ref{sec:The-procedure} we present an algorithm attaching
to any pair of coprime integers $(n,d)$ with $0< d<n$ a non-degenerate
unitary solution $r_{(n,d)}$ of the AYBE (\ref{eq:AYBE in 3 variables}).
The idea behind this algorithm is the computation of certain triple
Massey products in the bounded derived category of coherent sheaves
$\mbox{D}^{b}(\mbox{Coh}(E))$ for a cuspidal cubic curve $E=V(y^{2}z-x^{3})\subset\mathbb{P}^{2}$.
In all examples computed so far, solutions produced by this algorithm
satisfy (\ref{eq:Laurent assumption}). Moreover, in all these examples
$\bar{r}_{0}$ has no infinitesimal symmetries, i.e. there is no non-trivial
$a\in\mathfrak{sl}_{n}(\mathbb{C})$ such that \emph{\[
\Bigl[\bar{r}_{0}(y_{1},y_{2}),a\otimes\mathds{1}+\mathds{1}\otimes a\Bigr]=0.\]
}As observed by Polishchuk \cite{Polischuckb}, certain unitary solutions
of (\ref{eq:AYBE in 3 variables}) are closely related with the quantum
Yang-Baxter equation. The remaining sections of this paper are dedicated
to the generalization of his results. Our main result is the following:

\begin{thm}
\label{thm: generalizing Polishchuk}Let $r(v;y_{1},y_{2})$ be a
non-degenerate unitary solution of the AYBE (\ref{eq:AYBE in 3 variables})
of the form (\ref{eq:Laurent assumption}). If \emph{$\bar{r}_{0}(y_{1},y_{2})=(\mbox{pr}\otimes\mbox{pr})\left(r_{0}(y_{1},y_{2})\right)$}
has no infinitesimal symmetries, then the following hold.\\

\noindent i) For fixed $v_{0}\in\mathbb{C}^{\times}$, $\tilde{r}(y_{1},y_{2})=r(v_{0};y_{1},y_{2})$
is a solution of the quantum Yang-Baxter equation (QYBE) \begin{equation}
\tilde{r}^{12}(y_{1},y_{2})\,\tilde{r}^{13}(y_{1},y_{3})\,\tilde{r}^{23}(y_{2},y_{3})=\tilde{r}^{23}(y_{2},y_{3})\,\tilde{r}^{13}(y_{1},y_{3})\,\tilde{r}^{12}(y_{1},y_{2}).\label{eq:QYBE}\end{equation}

\noindent ii) Let $s(v;y_{1},y_{2})$ be another non-degenerate unitary
solution of the AYBE (\ref{eq:AYBE in 3 variables}) of the form (\ref{eq:Laurent assumption})
with \emph{$(\mbox{pr}\otimes\mbox{pr})\left(s_{0}\right)=\bar{r}_{0}$.}
Then there exists a meromorphic function $g:\mathbb{C}\rightarrow\mathbb{C}$
such that \emph{\[
s(v;y_{1},y_{2})=\exp\Bigl(v\left(g(y_{2})-g(y_{1})\right)\Bigr)r(v;y_{1},y_{2}).\]
}\noindent If additionally both $r$ and $s$ can be obtained by the procedure
described in section \ref{sec:The-procedure}, then $g$ is holomorphic.
Thus the above equation states that $r$ and $s$ are gauge equivalent
in that case.
\end{thm}
\noindent This result is proved in two steps, which are theorem \ref{thm: CYBE -> QYBE -> CYBE}
respectively theorem \ref{thm:recovering AYBE sol. from their CYBE}
and corollary \ref{cor: if r and s are both constructed geometrically}.
Let us add the following remarks:\\

\noindent $\bullet$ Theorem \ref{thm: generalizing Polishchuk} was
shown by Polishchuk, see \cite[Theorem 1.4]{Polischuckb} and \emph{\cite[Theorem 6]{Polischuck2002}},
in the case when $r(v;y_{1},y_{2})$ depends only on $v$ and the
difference $y=y_{1}-y_{2}$. However, solutions obtained by the algorithm
presented in section \ref{sec:The-procedure} do not have this property.

\noindent $\bullet$ The proof of theorem \ref{thm: generalizing Polishchuk}
is purely analytical.

\noindent $\bullet$ The solutions $\bar{r}_{0}$ of the CYBE (\ref{eq:CYBE})
obtained by the procedure presented in section \ref{sec:The-procedure}
belong to the class of \emph{rational} solutions. Rational solutions
of the CYBE (\ref{eq:CYBE}) have been classified by Stolin \cite{Stolin1991}.
As we shall show in a subsequent paper \cite{Henrichd}, the solutions
$\bar{r}_{0}$ produced by our algorithm form a subclass of rational
solutions which can be intrinsically described in terms of Stolin's
classification. One of the benefits of our method is that theorem
\ref{thm: generalizing Polishchuk} gives an explicit way of lifting
these solutions of the CYBE to solutions of the QYBE.

\noindent $\bullet$ There are solutions $r$ of (\ref{eq:AYBE in 3 variables})
which satisfy (\ref{eq:QYBE}) even though $\bar{r}_{0}$ has infinitesimal
symmetries. An example is given by the following function, depending
only on the difference $y=y_{1}-y_{2}$:\[
r(v,y)=\frac{1}{2v}\mathds{1}\otimes\mathds{1}+\frac{1}{y}\left(e_{11}\otimes e_{11}+e_{22}\otimes e_{22}+e_{12}\otimes e_{21}+e_{21}\otimes e_{12}\right).\]
 The corresponding solution of the CYBE (\ref{eq:CYBE}) is the rational
solution of Yang \[
\bar{r}_{0}(v)=\frac{1}{y}\left(\frac{1}{2}h\otimes h+e_{12}\otimes e_{21}+e_{21}\otimes e_{12}\right)\in\mathfrak{sl}_{2}(\mathbb{C})\otimes\mathfrak{sl}_{2}(\mathbb{C})\]
 whose infinitesimal symmetries are given by all non-zero elements
of $\mathfrak{sl}_{2}(\mathbb{C})$. Hence, there might exist a generalization
of theorem \ref{thm: generalizing Polishchuk} i).

\section{The algorithm\label{sec:The-procedure}}

In this section we present an algorithm which takes as input a pair
of coprime numbers $(n,d)$ with $0<d<n$ and produces a non-degenerate
unitary solution $r_{(n,d)}$ of the AYBE (\ref{eq:AYBE in 3 variables})
with values in $A\otimes A$ where $A=\mbox{Mat}_{n\times n}(\mathbb{C})$.
The algorithm is due to Burban and Kreu\ss{}ler \cite[Section 10]{Burban},
see in particular \cite[Algorithm 10.7]{Burban}. The actual content
of the described procedure is the computation of certain triple Massey
products in the bounded derived category of coherent sheaves $\mbox{D}^{b}(\mbox{Coh}(E))$
for a cuspidal cubic curve $E=V(y^{2}z-x^{3})\subset\mathbb{P}^{2}$.\\
 \\
 \textbf{$\star$ Step 1: construction of the matrix $J=J(n-d,d)$.}

\noindent We introduce the following map defined on all tuples of
coprime integers $(a,b)\neq(1,1)$:\[
\epsilon(a,b)=\left\{ \begin{tabular}{cc}
 $(a-b,b),$  &  $a>b$\\
$(a,b-a),$  &  $a<b$ \end{tabular}\right.\]
By assumption $(n,d)$ is a tuple of coprime integers. Hence it induces
a finite sequence of tuples ending with $(1,1)$, defined as follows.
We put $(a_{0},b_{0})=(n-d,d)$ and, as long as $(a_{i},b_{i})\neq(1,1)$,
we set $(a_{i+1},b_{i+1})=\epsilon(a_{i},b_{i})$. Next, let\[
J(1,1)=\left({\begin{tabular}{c|c}
 $0$  &  $1$ \\
\hline $0$  &  $0$ \end{tabular}}\right)\in\mbox{Mat}_{2\times2}(\mathbb{C}).\]
Assuming \[
J(a,b)=\left({\begin{tabular}{c|c}
 $J_{1}$  &  $J_{2}$ \\
\hline $0$  &  $J_{3}$ \end{tabular}}\right)\]
 with $J_{1}\in\mbox{Mat}_{a\times a}(\mathbb{C})$ and $J_{3}\in\mbox{Mat}_{b\times b}(\mathbb{C})$
has already been defined and that $(a,b)=\epsilon(p,q)$, we set\[
J(p,q)=\left\{ \begin{tabular}{cc}
 $\left({\begin{tabular}{c|cc}
 0  &  $\mathds{1}$  &  0 \\
\hline 0  &  $J_{1}$  &  $J_{2}$ \\
0  &  0  &  $J_{3}$ \end{tabular}}\right),$  &  $p=a$\\
\\$\left({\begin{tabular}{cc|c}
 $J_{1}$  &  $J_{2}$  &  0 \\
0  &  $J_{3}$  &  $\mathds{1}$ \\
\hline 0  &  0  &  0 \end{tabular}}\right),$  &  $q=b.$ \end{tabular}\right.\]

\noindent Hence, to $(n,d)$ we may associate the $n\times n$ matrix
$J=J(n-d,d)$ that is obtained from the matrix $J(1,1)$ and the sequence
$\bigl\{(n-d,d),...,(1,1)\bigr\}$ by applying the recursive procedure
described above.

\begin{example} \emph{
Let $(n,d)=(5,2)$. The induced sequence is $\bigl\{(3,2),(1,2),(1,1)\bigr\}$
and $J=J(3,2)$ is constructed as follows\emph{\[
\left({\begin{tabular}{c|c}
 $0$  &  $1$ \\
\hline $0$  &  $0$ \end{tabular}}\right)\rightarrow\left({\begin{tabular}{c|cc}
 $0$  &  $1$  &  $0$\\
\hline $0$  &  $0$  &  $1$\\
$0$  &  $0$  &  $0$ \end{tabular}}\right)\rightarrow\left({\begin{tabular}{ccc|cc}
 $0$  &  $1$  &  $0$  &  $0$  &  $0$ \\
$0$  &  $0$  &  $1$  &  $1$  &  $0$ \\
$0$  &  $0$  &  $0$  &  $0$  &  $1$ \\
\hline $0$  &  $0$  &  $0$  &  $0$  &  $0$ \\
$0$  &  $0$  &  $0$  &  $0$  &  $0$ \end{tabular}}\right).\]
}  }
\end{example}
\noindent \textbf{$\star$ Step 2: definition of $\mbox{Sol}_{n,d}^{v,y_{1}}$,
$\mbox{res}_{y_{1}}$ and $\mbox{ev}_{y_{2}}$.}

\noindent For the partition of $A=\mbox{Mat}_{n\times n}(\mathbb{C})$
induced by the partition of $J$ as in step 1, we introduce the following
subspace of the polynomial ring $A[z]$: \[
W_{n,d}=\left\{ F(z)=\left(\begin{tabular}{c|c}
 $W$  &  $X$ \\
\hline $Y$  &  $Z$ \end{tabular}\right)+\left(\begin{tabular}{c|c}
 $W'$  &  $0$ \\
\hline $Y'$  &  $Z'$ \end{tabular}\right)z+\left(\begin{tabular}{c|c}
 $0$  &  $0$ \\
\hline $Y''$  &  $0$ \end{tabular}\right)z^{2}\right\} .\]
Next, for $F(z)\in W_{n,d}$ we denote\[
F_{0}=\left(\begin{tabular}{c|c}
 $W'$  &  $X$ \\
\hline $Y''$  &  $Z'$ \end{tabular}\right)\mbox{ and }\, F_{\epsilon}=\left({\begin{tabular}{c|c}
 $W$  &  $0$ \\
\hline $Y'$  &  $Z$ \end{tabular}}\right).\]
Then for $v,y_{1}\in\mathbb{C}$, we define the following subspace
of $W_{n,d}$:\[
\mbox{Sol}_{n,d}^{v,y_{1}}=\Bigl\{ F(z)\in W_{n,d}\,\Bigl|\,\left[F_{0},J\right]+(y_{1}-v)F_{0}+F_{\epsilon}=0\Bigr\}.\]

\begin{prop}
\cite[Section 10]{Burban}\label{pro:res_y1 : Sol -> sln iso} The
vector space \emph{$\mbox{Sol}_{n,d}^{v,y_{1}}$} has dimension $n^{2}$
and for $y_{1}\neq y_{2}\in\mathbb{C}$\emph{, $\mbox{ev}_{y_{2}}:\mbox{Sol}_{n,d}^{v,y_{1}}\rightarrow A$}
defined by \emph{$\mbox{ev}_{y_{2}}\bigl(F(z)\bigr)=\frac{1}{y_{2}-y_{1}}F(y_{2})$}
is an isomorphism. For $v\neq0$, \emph{$\mbox{res}_{y_{1}}:\mbox{Sol}_{n,d}^{v,y_{1}}\rightarrow A$}
given by \emph{$\mbox{res}_{y_{1}}\bigl(F(z)\bigr)=F(y_{1})$} is
an isomorphism as well. 
\end{prop}
We will assume $v\neq0$ and $y_{1}\neq y_{2}$ for the rest of this
section. Then we get a linear automorphism $\tilde{r}_{(n,d)}$ of
the matrix algebra $A$ given by the formula $\tilde{r}_{(n,d)}(v;y_{1},y_{2})=\mbox{ev}_{y_{2}}\circ\mbox{res}_{y_{1}}^{-1}$.\\

\noindent \textbf{$\star$ Step 3: definition of the tensor $r_{(n,d)}(v;y_{1},y_{2})$.}

\noindent Note that we have a canonical isomorphism of vector spaces\[
\mbox{can}:A\otimes A\rightarrow\mbox{End}_{\mathbb{C}}(A),\, X\otimes Y\mapsto\bigl(Z\mapsto\mbox{tr}(XZ)Y\bigr).\]
 For fixed $v,y_{1},y_{2}$, we set $r_{(n,d)}(v;y_{1},y_{2})=\mbox{can}^{-1}\bigl(\tilde{r}_{(n,d)}(v;y_{1},y_{2})\bigr)$.
The proof of the following theorem is contained in \cite[Section 10]{Burban}.

\begin{thm}
\label{thm: r(n,d) is unitary non-deg AYBE solution}The tensor-valued
function $r_{(n,d)}:\left(\mathbb{C}_{(v;y_{1},y_{2})}^{3},0\right)\rightarrow A\otimes A$
is a non-degenerate unitary solution of (\ref{eq:AYBE in 3 variables}).
Moreover $r_{(n,d)}(v;y_{1},y_{2})$ is holomorphic on $\bigl(\mathbb{C}^{3}\setminus V\bigl(v(y_{1}-y_{2})\bigr)\bigr)$.
\\
 
\end{thm}
\begin{example}
\emph{ \label{exa: examples for r_2,1 and r_3,1}For any $n\in\mathbb{N}$,
let $P=\sum_{1\leq i,j\leq n}e_{ij}\otimes e_{ji}\in A\otimes A$.\emph{
}\\
\emph{ }\\
\emph{ } i) Let $(n,d)=(2,1)$. Then we have }

\noindent \emph{\[
r_{(2,1)}(v;y_{1},y_{2})=\frac{1}{2v}\mathds{1}\otimes\mathds{1}+\frac{1}{y_{2}-y_{1}}P+\]
 \[
+\left(v-y_{1}\right)e_{21}\otimes\check{h}+\left(v+y_{2}\right)\check{h}\otimes e_{21}-\frac{v(v-y_{1})(v+y_{2})}{2}e_{21}\otimes e_{21},\]
 } \emph{where $\check{h}=\mbox{diag}(\frac{1}{2},-\frac{1}{2})$.\emph{} }\\

\noindent \emph{ ii) Let $(n,d)=(3,1)$. Then we have\emph{ \[
r_{(3,1)}(v;y_{1},y_{2})=\frac{1}{3v}\mathds{1}\otimes\mathds{1}+\frac{1}{y_{2}-y_{1}}P-\]
 \[
-e_{21}\otimes\check{h}_{1}+\check{h}_{1}\otimes e_{21}+e_{32}\otimes e_{12}-e_{12}\otimes e_{32}-y_{1}e_{32}\otimes\check{h}_{2}+y_{2}\check{h}_{2}\otimes e_{32}+\]
 \[
+(v-y_{1})e_{31}\otimes e_{12}+(v+y_{2})e_{12}\otimes e_{31}+ve_{32}\otimes(e_{11}-e_{33})+v(e_{11}-e_{33})\otimes e_{32}+\]
 \[
+\frac{1}{3}v(y_{1}-3v)e_{32}\otimes e_{21}+\frac{1}{3}v(y_{2}+3v)e_{21}\otimes e_{32}+v(v-y_{1})e_{31}\otimes\check{h}_{1}-\]
 \[
-v(v+y_{2})\check{h}_{1}\otimes e_{31}+\frac{2}{3}v^{2}(y_{1}-v)e_{31}\otimes e_{21}-\frac{2}{3}v^{2}(y_{2}+v)e_{21}\otimes e_{31}+\]
 \[
+\frac{1}{3}v^{2}(y_{2}+v)(3v-y_{1})e_{32}\otimes e_{31}+\frac{1}{3}v^{2}(y_{1}+v)(3v+y_{2})e_{31}\otimes e_{32}+\]
 \[
+\frac{2}{3}ve_{21}\otimes e_{21}+\frac{2}{3}v^{3}(v-y_{1})(v+y_{2})e_{31}\otimes e_{31}\]
 \[
+\frac{1}{3}v(-6v^{2}+3v(y_{1}-y_{2})+2y_{1}y_{2})e_{32}\otimes e_{32},\]
 }where $\check{h}_{1}=\mbox{diag}(\frac{2}{3},-\frac{1}{3},-\frac{1}{3})$
and $\check{h}_{2}=\mbox{diag}(\frac{1}{3},\frac{1}{3},-\frac{2}{3})$.\emph{
 } }
\end{example}
\noindent Next, recall \cite[Lemma 2.11]{Burban}, which is based
upon \cite[Lemma 1.2]{Polischuck2002}:

\begin{lem}
\emph{\label{lem: AYBE unitary ->CYBE unitary }} Let $r(v;y_{1},y_{2})$
be a unitary solution of (\ref{eq:AYBE in 3 variables}) of the form
(\ref{eq:Laurent assumption}). Then \emph{$\bar{r}_{0}(y_{1},y_{2})=(\mbox{pr}\otimes\mbox{pr})\left(r_{0}(y_{1},y_{2})\right)$}
is a unitary solution of the CYBE (\ref{eq:CYBE}). 
\end{lem}
\noindent Combining this result with example (\ref{exa: examples for r_2,1 and r_3,1}),
we derive that $r_{(2,1)}$ and $r_{(3,1)}$ induce solutions of of
the CYBE (\ref{eq:CYBE}). Let us denote these by $c_{(2,1)}$ and
$c_{(3,1)}$ respectively. Also, let $\Omega$ be the Casimir element
of $\mathfrak{sl}_{n}(\mathbb{C})\otimes\mathfrak{sl}_{n}(\mathbb{C})$
with respect to the trace form $(x,y)\mapsto\mbox{tr}(x\cdot y)$
\[
\Omega=\sum_{1\leq i\neq j\leq n}e_{i,j}\otimes e_{j,i}+\sum_{1\leq l\leq n-1}h_{l}\otimes\check{h}_{l}.\]
 Observing that $\left(\mbox{pr}\otimes\mbox{pr}\right)\left(P\right)=\Omega$,
we derive the following formulae:

\[
c_{(2,1)}(y_{1},y_{2})=\frac{\Omega}{y_{2}-y_{1}}+y_{2}\check{h}\otimes e_{21}-y_{1}e_{21}\otimes\check{h}\in\mathfrak{sl}_{2}(\mathbb{C})\otimes\mathfrak{sl}_{2}(\mathbb{C})\]
 and\[
c_{(3,1)}(y_{1},y_{2})=\frac{\Omega}{y_{2}-y_{1}}+y_{2}\check{h}_{2}\otimes e_{32}-y_{1}e_{32}\otimes\check{h}_{2}+y_{2}e_{12}\otimes e_{31}-y_{1}e_{31}\otimes e_{12}-\]
 \[
-e_{21}\otimes\check{h}_{1}+\check{h}_{1}\otimes e_{21}+e_{32}\otimes e_{12}-e_{12}\otimes e_{32}\in\mathfrak{sl}_{3}(\mathbb{C})\otimes\mathfrak{sl}_{3}(\mathbb{C}).\]
 It can be verified that neither $c_{(2,1)}$ nor $c_{(3,1)}$ has
any infinitesimal symmetries. Thus theorem \ref{thm: generalizing Polishchuk}
1) yields that for fixed $v_{0}\in\mathbb{C}^{\times}$, both $r_{(2,1)}(v_{0};y_{1},y_{2})$
and $r_{(3,1)}(v_{0};y_{1},y_{2})$ satisfy the QYBE (\ref{eq:QYBE}).

\section{Gauge equivalence and general results on the AYBE}

In this section we explain the notion of gauge equivalence and collect
some useful results for solutions of (\ref{eq:AYBE in 3 variables}).

In order to deal with gauge equivalences in a correct way, we have
to consider a more general form of the AYBE in four variables\begin{equation}
\begin{array}{c}
r^{12}(v_{1},v_{2};y_{1},y_{2})\, r^{23}(v_{1},v_{3};y_{2},y_{3})=\\
=r^{13}(v_{1},v_{3};y_{1},y_{3})\, r^{12}(v_{3},v_{2};y_{1},y_{2})+r^{23}(v_{2},v_{3};y_{2},y_{3})\, r^{13}(v_{1},v_{2};y_{1},y_{3})\end{array}\label{eq:AYBE in 4 var}\end{equation}
 where $r$ now denotes the germ of a meromorphic function $r:\left(\mathbb{C}^{4},0\right)\rightarrow A\otimes A$.
From the point of view of algebraic geometry, (\ref{eq:AYBE in 4 var})
is more natural than (\ref{eq:AYBE in 3 variables}), see \cite{Burban}.
Note that solutions of (\ref{eq:AYBE in 3 variables}) are just solutions
of (\ref{eq:AYBE in 4 var}) depending on the difference of the first
pair of spectral parameters $r(v_{1},v_{2};y_{1},y_{2})=r(v_{1}-v_{2};y_{1},y_{2})=r(v;y_{1},y_{2})$.
As to the definition of gauge equivalence, this is given as follows:

\begin{defn}
\emph{ \label{def:gauge equivalence on AYBE}Let $\phi:(\mathbb{C}^{2},0)\rightarrow\mbox{GL}_{n}(\mathbb{C})$
be the germ of a holomorphic function and let $r(v_{1},v_{2};y_{1},y_{2})$
be a solution of (\ref{eq:AYBE in 4 var}). Then the tensor valued
function \[
r'(v_{1},v_{2};y_{1},y_{2})=\Bigl(\phi(v_{1};y_{1})\otimes\phi(v_{2};y_{2})\Bigr)\, r(v_{1},v_{2};y_{1},y_{2})\,\Bigl(\phi^{-1}(v_{2};y_{1})\otimes\phi^{-1}(v_{1};y_{2})\Bigr)\]
 is also a solution of (\ref{eq:AYBE in 4 var}). The solutions $r$
and $r'$ are said to be gauge equivalent and $\phi$ is called a
gauge transformation.\emph{ }}
\end{defn}
\begin{example}
\emph{ \label{exa: example for gauge equivalence}Let $r(v_{1},v_{2};y_{1},y_{2})\in A\otimes A$
be a solution of (\ref{eq:AYBE in 4 var}), $c\in\mathbb{C}$ and
$\phi(v,y)=\exp(cvy)\cdot\mathds{1}:(\mathbb{C}^{2},0)\rightarrow\mbox{GL}_{n}(\mathbb{C})$
be a gauge transformation. Then \[
\exp\Bigl(c\left(v_{2}-v_{1}\right)\left(y_{2}-y_{1}\right)\Bigr)\, r(v_{1},v_{2};y_{1},y_{2})\]
 is a solution of (\ref{eq:AYBE in 4 var}), gauge equivalent to $r$.Ê}

\emph{ Similarly, assume that $r(v_{1},v_{2};y_{1},y_{2})=r(v;y_{1},y_{2})$
is a solution of (\ref{eq:AYBE in 4 var}) which depends only on $v=v_{1}-v_{2},y_{1},y_{2}$.
Thus $r(v;y_{1},y_{2})$ is a solution of (\ref{eq:AYBE in 3 variables}).
Consider the gauge transformation $\phi(v,y)=\exp(vg(y))\cdot\mathds{1}:(\mathbb{C}^{2},0)\rightarrow\mbox{GL}_{n}(\mathbb{C})$
for some holomorphic function $g:\mathbb{C}\rightarrow\mathbb{C}$.
Then $r'(v;y_{1},y_{2})=\exp\bigl(v\left(g(y_{2})-g(y_{1})\right)\bigr)\, r(v;y_{1},y_{2})$
is a solution of (\ref{eq:AYBE in 3 variables}) as well.\emph{ }}
\end{example}
\noindent In the remainder of this section, we list some basic results
on solutions of the AYBE which we shall need in the next sections.

\begin{lem}
\emph{\label{'Dual' AYBE }\cite[Lemma 2.7]{Burban}} Let $r(v_{1},v_{2};y_{1},y_{2})$
be a unitary solution of (\ref{eq:AYBE in 4 var}). Then writing $r^{ij}(v_{1},v_{2})$
as short-hand for $r^{ij}(v_{1},v_{2};y_{i},y_{j})$, $r$ also satisfies
the {}``dual equation''\[
r^{23}(v_{2},v_{3})\, r^{12}(v_{1},v_{3})=r^{12}(v_{1},v_{2})\, r^{13}(v_{2},v_{3})+r^{13}(v_{1},v_{3})\, r^{23}(v_{2},v_{1}).\]

\end{lem}
\begin{cor}
If $r(v;y_{1},y_{2})$ is a unitary solution of (\ref{eq:AYBE in 3 variables}),
then we also have\begin{equation}
\begin{array}{c}
{\displaystyle \frac{{}}{{}}}r^{23}(u+v;y_{2},y_{3})\, r^{12}(u;y_{1},y_{2})=\\
{\displaystyle \frac{{}}{{}}}=r^{12}(-v;y_{1},y_{2})\, r^{13}(u+v;y_{1},y_{3})+r^{13}(u;y_{1},y_{3})\, r^{23}(v;y_{2},y_{3}).\end{array}\label{eq:Dual AYBE in 3 variables}\end{equation}

\end{cor}
\noindent The proof of the next lemma is essentially contained in
the proof of \cite[Theorem 5]{Polischuck2002}.

\begin{lem}
\label{lem: r uniquely determined by r_0,r_1}Let $r(v;y_{1},y_{2})$
be a unitary solution of (\ref{eq:AYBE in 3 variables}) of the form
(\ref{eq:Laurent assumption}). Then $r$ is uniquely determined by
$r_{0}$ and $r_{1}$. Moreover, we have\begin{equation}
\begin{array}{c}
{\displaystyle \frac{{}}{{}}}r_{1}^{12}(y_{1},y_{2})+r_{1}^{13}(y_{1},y_{3})+r_{1}^{23}(y_{2},y_{3})=\\
{\displaystyle \frac{{}}{{}}}=r_{0}^{12}(y_{1},y_{2})\, r_{0}^{13}(y_{1},y_{3})-r_{0}^{23}(y_{2},y_{3})\, r_{0}^{12}(y_{1},y_{2})+r_{0}^{13}(y_{1},y_{3})\, r_{0}^{23}(y_{2},y_{3}).\end{array}\label{eq:r1 determined by r0}\end{equation}

\end{lem}
\begin{proof}
First we show that $r$ is uniquely determined by $r_{0},r_{1}$ and
$r_{2}$. To this end, we fix $k>2$ and show how to construct $r_{k}$
from $\left\{ r_{i}\right\} _{0\leq i\leq k-1}$. Let us insert the
Laurent expansion (\ref{eq:Laurent assumption}) of $r$ into (\ref{eq:Dual AYBE in 3 variables})
and examine the terms of total degree $k-1$ in the variables $u$
and $v$. We derive the equation\begin{equation}
\begin{array}{c}
r_{k}^{12}(y_{1},y_{2})\,\left[{\displaystyle \frac{(-v)^{k}}{u+v}-\frac{u^{k}}{u+v}}\right]+r_{k}^{13}(y_{1},y_{3})\,\left[{\displaystyle \frac{u^{k}}{v}-\frac{(u+v)^{k}}{v}}\right]+\\
+r_{k}^{23}(y_{2},y_{3})\,\left[{\displaystyle \frac{v^{k}}{u}-\frac{(u+v)^{k}}{u}}\right]=\dots\end{array}\label{eq:recovering_AYBE_from_CYBE_1}\end{equation}
 where the ride-hand side contains terms $r_{i}$ with $i<k$ only.
The polynomials in $u$ and $v$ on the left-hand side are linearly
independent for $k>2$. Indeed, if we place everything over a common
denominator and focus on the coefficients of $u^{k}$ in the respective
terms \[
u(-v)^{k+1}-vu^{k+1},\, u^{k+1}(u+v)-u(u+v)^{k+1},\,(u+v)v^{k+1}-v(u+v)^{k+1}\]
 then these are $-vu,\, u(u+v)-{k+1 \choose 2}v^{2}$ and $-\left(k+1\right)v^{2}$
respectively. This proves our claim that $r$ is determined by the
$r_{k}$ with $k\leq2$.

In the next step, we show that $r_{2}$ is already determined by $r_{0}$
and $r_{1}$. Indeed, for $k=2$ equation (\ref{eq:recovering_AYBE_from_CYBE_1})
reads \[
\left(v-u\right)\, r_{2}^{12}(y_{1},y_{2})\,-\left(2u+v\right)\, r_{2}^{13}(y_{1},y_{3})-\left(u+2v\right)r_{2}^{23}(y_{2},y_{3})=\dots\]
 that is \[
-u\cdot\left(r_{2}^{12}(y_{1},y_{2})+2r_{2}^{13}(y_{1},y_{3})+r_{2}^{23}(y_{2},y_{3})\right)+\]
 \[
+v\cdot\left(r_{2}^{12}(y_{1},y_{2})-r_{2}^{13}(y_{1},y_{3})-2r_{2}^{23}(y_{2},y_{3})\right)=\dots\]
 with the the right-hand side depending on $r_{0}$ and $r_{1}$ only.
Let us denote the coefficient of $-u$ on the left-hand side by $a$,
that of $v$ by $b$. Since $a,b$ are determined by $r_{0}$ and
$r_{1}$ only, so is $\frac{a+2b}{3}=r_{2}^{12}(y_{1},y_{2})-r_{2}^{23}(y_{2},y_{3})$.
Thus, $r_{2}(y_{1},y_{2})$ is determined by $r_{0}$ and $r_{1}$.
Putting $k=1$ in (\ref{eq:recovering_AYBE_from_CYBE_1}), we obtain
(\ref{eq:r1 determined by r0}). 
\end{proof}

\section{Poles of solutions of the AYBE}

In this section, we study the poles of solutions of (\ref{eq:AYBE in 3 variables})
along $y_{1}=y_{2}$. We start with the following easy fact on $P=\sum_{1\leq i,j\leq n}e_{i,j}\otimes e_{j,i}\in A\otimes A$.

\begin{fact} \emph{
\label{fac:P is the center}\emph{Any tensor $\theta\in A\otimes A$
such that $\theta(x\otimes1)=(1\otimes x)\theta$ for all $x\in A$
is a scalar multiple of $P$. Moreover $P(1\otimes x)=(x\otimes1)P$
for any $x\in A$. } }
\end{fact}
\begin{lem}
\emph{\label{lem:poles are simple}\cite[Lemma 1.3]{Polischuckb}}
Let $r(u;y_{1},y_{2})$ be a non-degenerate unitary solution of (\ref{eq:AYBE in 3 variables}).
Assume that $r(u;y_{1},y_{2})$ has a pole along $y_{1}=y_{2}$. Then
this pole is simple and \emph{$\mbox{lim}_{y_{2}\rightarrow y_{1}}(y_{1}-y_{2})\, r(u;y_{1},y_{2})=c\cdot P$}
for some $c\in\mathbb{C}$. 
\end{lem}
\begin{proof}
Write $r(u;y_{1},y_{2})=\alpha(u;y_{1}-y_{2})+\beta(u;y_{1},y_{2})$
and assume that no summand of $\beta(u;y_{1},y_{2})$ depends only
on $u$ and $y=y_{1}-y_{2}$. Let $\alpha(u;y)=\frac{\theta(u)}{y^{k}}+\frac{\eta(u)}{y^{k-1}}+\dots$
be the Laurent expansion near $y=0$. In order to see that $k\leq1$,
we consider the polar parts in (\ref{eq:AYBE in 3 variables}) as
$y_{3}\rightarrow y_{1}$, which yields\begin{equation}
\theta^{13}(u+v)\, r^{12}(-v;y_{1},y_{2})+r^{23}(v;y_{2},y_{1})\,\theta^{13}(u)=0.\label{eq:gen_massey_on_cycles_1}\end{equation}
 Analogously, for $y_{2}\rightarrow y_{1}$ \begin{equation}
\theta^{12}(u)\, r^{23}(u+v;y_{1},y_{3})-r^{13}(u+v;y_{1},y_{3})\,\theta^{12}(-v)=0.\label{eq:gen_massey_on_cycles_2}\end{equation}
 Let $V\subseteq A$ be the minimal subspace such that $\theta(u)\in V\otimes A$
for all $u$ where $\theta(u)$ is defined. Obviously $r^{23}(u;y_{1},y_{2})\,\theta^{13}(u)\in V\otimes A\otimes A$
hence by (\ref{eq:gen_massey_on_cycles_1}) $\theta^{13}(u+v)\, r^{12}(-v;y_{1},y_{2})\in V\otimes A\otimes A$
as well. Thus, $r^{12}(u;y_{1},y_{2})\in A_{1}\otimes A$, where\[
A_{1}=\left\{ \left.a\in A\right|\theta(u)\left(a\otimes1\right)\in V\otimes A\,\mbox{for all }u\right\} .\]
 By non-degeneracy $A_{1}=A$, thus $VA\subseteq V$. Similarly, using
(\ref{eq:gen_massey_on_cycles_2}), we get $AV\subseteq V$, so that
$V$ is a two-sided non-zero ideal in $A$. Hence $V=A$.

Let us come back to (\ref{eq:AYBE in 3 variables}). We want to have
a look at the coefficient of $(y_{1}-y_{2})^{1-k}$ in the expansion
of (\ref{eq:AYBE in 3 variables}) near $y_{2}-y_{1}=h$ equal to
zero. The terms contributing to this only depend on $r^{12}(u;y_{1},y_{1}+h)\, r^{23}(u+v;y_{1}+h,y_{3})$
and $r^{13}(u+v;y_{1},y_{3})\, r^{12}(-v;y_{1},y_{1}+h)$. Thus the
coefficient of $h^{1-k}$ consists of two summands, the first one
being $\eta^{12}(u)\, r^{23}(u+v;y_{1},y_{3})-r^{13}(u+v;y_{1},y_{3})\,\eta^{12}(-v)$
and second one being $\theta^{12}(u)$ times the coefficient of $h$
in $r^{23}(u+v;y_{1}+h,y_{3})$. Note that for this last summand to
be non-zero we must assume $k>1$. Now $r^{23}(u+v;y_{1}+h,y_{3})-r^{23}(u+v;y_{1},y_{3})$
equals the summand of $r^{23}(u+v;y_{1}+h,y_{3})$ divisible by $h$,
hence the coefficient of $h^{1-k}$ is exactly \[
\eta^{12}(u)\, r^{23}(u+v;y_{1},y_{3})-r^{13}(u+v;y_{1},y_{3})\,\eta^{12}(-v)+\theta^{12}(u)\,\frac{\partial r^{23}}{\partial y_{1}}(u+v;y_{1},y_{3}).\]
 Examining the polar parts in the above expression for $y_{1}-y_{3}$
in a neighborhood of zero, we deduce that $\theta^{12}(u)\,\theta^{23}(u+v)=0$.
Setting $v=0$ this amounts to saying that $\theta(u)=\left\{ \left.a\otimes b\right|ab=0\right\} $.
Since $V=A$ this is a contradiction. Therefore $k=1$.\\

\noindent Next, we have a look at the polar parts in (\ref{eq:Dual AYBE in 3 variables})
near $y_{3}=y_{2}$. We deduce $\theta^{23}(u+v)\, r^{12}(u;y_{1},y_{2})=r^{13}(u;y_{1},y_{2})\,\theta^{23}(v)$.
Hence $r(u;y_{1},y_{2})\in A\otimes A(u)$, where \[
A(u)=\left\{ a\in A\left|\theta(u+v)(x\otimes1)=(1\otimes x)\theta(v)\,\mbox{for all}\, v\right.\right\} .\]
 Since $A$ is non-degenerate this implies $A(u)=A$ for generic $u$,
in which case $\mathds{1}\in A(u)$ and thus $\theta(u+v)=\theta(u)$.
Hence $\theta=\theta(0)$ is constant. Recalling fact \ref{fac:P is the center}
finishes the proof. 
\end{proof}
\begin{cor}
\emph{\label{cor:simple poles exist}\cite[Lemma 1.5]{Polischuckb}}
Let $r(u;y_{1},y_{2})$ be a non-degenerate unitary solution of (\ref{eq:AYBE in 3 variables})
of the form (\ref{eq:Laurent assumption}). Then $r(u;y_{1},y_{2})$
has a simple pole along $y_{1}=y_{2}$ with residue a scalar multiple
of $P$. 
\end{cor}
\begin{proof}
This is essentially the same proof as that of lemma 1.5 in \cite{Polischuckb},
using lemma \ref{lem:poles are simple} where Polishchuk refers to
lemma 1.3 of his paper. 
\end{proof}

\section{Quantization of solutions of CYBE coming from solutions of AYBE}

In this section we prove part i) of theorem \ref{thm: generalizing Polishchuk}:

\begin{thm}
\emph{\label{thm: CYBE -> QYBE -> CYBE}\cite[Theorem 1.4]{Polischuckb}}
Let $r(u;y_{1},y_{2})$ be a non-degenerate unitary solution of (\ref{eq:AYBE in 3 variables})
of the form (\ref{eq:Laurent assumption}) and let \emph{$\overline{r}_{0}(y_{1},y_{2})=\left(\mbox{pr}\otimes\mbox{pr}\right)\left(r_{0}(y_{1},y_{2})\right)$}. 
\end{thm}
\begin{enumerate}
\item \emph{$\overline{r}_{0}(y_{1},y_{2})$ is a non-degenerate unitary
solution of the CYBE (\ref{eq:CYBE}).} 
\item \emph{The following conditions are equivalent:}

\begin{enumerate}
\item \emph{for fixed $u\in\mathbb{C}^{\times}$, $r(u;y_{1},y_{2})$ satisfies
the QYBE (\ref{eq:QYBE}).} 
\item \emph{there exits a scalar function $\varphi(u;y_{1},y_{2})$ such
that \[
r(u;y_{1},y_{2})\, r(-u;y_{1},y_{2})=\varphi(u;y_{1},y_{2})\left(\mathds{1}\otimes\mathds{1}\right).\]
}  
\item \emph{for $i\in\{1,2\}$ there exists a scalar function $\psi_{i}(y_{1},y_{2})$
such that \[
\frac{\partial}{\partial y_{i}}\left(r_{0}(y_{1},y_{2})-\overline{r}_{0}(y_{1},y_{2})\right)=\psi_{i}(y_{1},y_{2})\left(\mathds{1}\otimes\mathds{1}\right).\]
}  
\item \emph{we have}\[
(\mbox{pr}\otimes\mbox{pr}\otimes\mbox{pr})\,\left[\overline{r}_{0}^{12}(y_{1},y_{2})\,\overline{r}_{0}^{13}(y_{1},y_{3})-\right.\]
 \textbf{\[
\left.-\overline{r}_{0}^{23}(y_{2},y_{3})\,\overline{r}_{0}^{12}(y_{1},y_{2})+\overline{r}_{0}^{13}(y_{1},y_{3})\,\overline{r}_{0}^{23}(y_{2},y_{3})\right]=0.\]
}  
\end{enumerate}
\item \emph{These conditions are satisfied if $\overline{r}_{0}(y_{1},y_{2})$
has no infinitesimal symmetries.}\\

\end{enumerate}
Before proving this statement, we first need to establish some auxiliary
results. The reader might wish to postpone checking them and to go
to the proof of theorem \ref{thm: CYBE -> QYBE -> CYBE} at the end
of this section immediately.

\begin{lem}
\emph{\label{lem:QYBE in terms of s*r-r*s}\cite[Lemma 1.6]{Polischuckb}}
For any triple of variables $u_{1},u_{2},u_{3}$ set $u_{ij}=u_{i}-u_{j}$.
Let $r(u;y_{1},y_{2})$ be any unitary solution of (\ref{eq:AYBE in 3 variables})
and $s(u;y_{1},y_{2})=r(u;y_{1},y_{2})\, r(-u;y_{1},y_{2})$. Then\[
r^{12}(u_{12};y_{1},y_{2})\, r^{13}(u_{23};y_{1},y_{3})\, r^{23}(u_{12};y_{2},y_{3})-\]
 \[
-r^{23}(u_{23};y_{2},y_{3})\, r^{13}(u_{12};y_{1},y_{3})\, r^{12}(u_{23};y_{1},y_{2})=\]
 \[
=s^{23}(u_{23};y_{2},y_{3})\, r^{13}(u_{13};y_{1},y_{3})-r^{13}(u_{13};y_{1},y_{3})\, s^{23}(u_{21};y_{2},y_{3})=\]
 \[
=r^{13}(u_{13};y_{1},y_{3})\, s^{12}(u_{32};y_{1},y_{2})-s^{12}(u_{12};y_{1},y_{2})\, r^{13}(u_{13};y_{1},y_{3}).\]

\end{lem}
\begin{proof}
Let us write $r^{ij}(u)$ as short-hand for $r^{ij}(u;y_{i},y_{j})$.
Since we may assume $u=u_{12}$, $v=u_{23}$ and $u+v=u_{13}$, (\ref{eq:AYBE in 3 variables})
may be written as\begin{equation}
r^{12}(u_{12})\, r^{23}(u_{13})=r^{13}(u_{13})\, r^{12}(u_{32})+r^{23}(u_{23})\, r^{13}(u_{12}).\label{eq:gen_massey_on_cycles_3}\end{equation}
 Analogously, putting $u=u_{13}$ and $v=u_{21}$, (\ref{eq:Dual AYBE in 3 variables})
reads\begin{equation}
r^{23}(u_{23})\, r^{12}(u_{13})=r^{12}(u_{12})\, r^{13}(u_{23})+r^{13}(u_{13})\, r^{23}(u_{21}).\label{eq:gen_massey_on_cycles_4}\end{equation}
 Multiplying (\ref{eq:gen_massey_on_cycles_4}) with $r^{23}(u_{12})$
from the right yields \[
r^{23}(u_{23})\, r^{12}(u_{13})\, r^{23}(u_{12})=r^{12}(u_{12})\, r^{13}(u_{23})\, r^{23}(u_{12})+r^{13}(u_{13})\, s^{23}(u_{21})\]
 while switching $u_{2}$ and $u_{3}$ in (\ref{eq:gen_massey_on_cycles_3})
followed by multiplication with $r^{23}(u_{23})$ from the left yields
\[
r^{23}(u_{23})\, r^{12}(u_{13})\, r^{23}(u_{12})=r^{23}(u_{23})\, r^{13}(u_{12})\, r^{12}(u_{23})+s^{23}(u_{23})\, r^{13}(u_{13}).\]
 Subtracting these equations, we end up with \[
r^{12}(u_{12})\, r^{13}(u_{23})\, r^{23}(u_{12})-r^{23}(u_{23})\, r^{13}(u_{12})\, r^{12}(u_{23})=\]
 \[
=s^{23}(u_{23})\, r^{13}(u_{13})-r^{13}(u_{13})\, s^{23}(u_{21}).\]
 Switching indices 1 and 3 and using unitarity of $r$ yields the
other identity. 
\end{proof}
\noindent For the next statement we need the notion of an infinitesimal
symmetry of a solution $r$ of (\ref{eq:AYBE in 3 variables}), which
is simply that of an element $a\in\mathfrak{sl}_{n}(\mathbb{C})$
such that $\left[r(u;y_{1},y_{2}),a^{1}+a^{2}\right]=0$, where $a^{1}=a\otimes\mathds{1}$
and $a^{2}=\mathds{1}\otimes a$.

\begin{lem}
\emph{\label{lem:form of s}\cite[Lemma 1.7]{Polischuckb}} Let $r(u;y_{1},y_{2})$
be a unitary solution of (\ref{eq:AYBE in 3 variables}) of the form
(\ref{eq:Laurent assumption}) and $s(u;y_{1},y_{2})=r(u;y_{1},y_{2})\, r(-u;y_{1},y_{2})$.
Assuming that $r(u;y_{1},y_{2})$ has a simple pole along $y_{1}=y_{2}$
with residue $cP$ for some $c\in\mathbb{C}$, we have\[
s(u;y_{1},y_{2})=a\otimes1+1\otimes a+\left(f(u)+g(y_{1},y_{2})\right)\mathds{1}\otimes\mathds{1}\]
 where $f(u)=f(-u)$, $g(y_{1},y_{2})=g(y_{2},y_{1})$ and \emph{$a\in\mathfrak{sl}_{n}(\mathbb{C})$}
is an infinitesimal symmetry of $r(u;y_{1},y_{2})$. Moreover, we
may write \[
r_{0}(y_{1},y_{2})=\overline{r}_{0}(y_{1},y_{2})+\alpha(y_{2})\otimes\mathds{1}-\mathds{1}\otimes\alpha(y_{1})+h(y_{1},y_{2})\mathds{1}\otimes\mathds{1}\]
 with \emph{$\overline{r}_{0}(y_{1},y_{2})$} mapping to \emph{$\mathfrak{sl}_{n}(\mathbb{C})\otimes\mathfrak{sl}_{n}(\mathbb{C})$,}
\emph{$\alpha(y)$} to $\mathfrak{sl}_{n}(\mathbb{C})$, $h(y_{1},y_{2})$
a scalar function and \[
\alpha(y)=\alpha(0)+\frac{y}{cn}a.\]

\end{lem}
\begin{proof}
By assumption $r(u;y_{1},y_{2})=\frac{c}{y_{1}-y_{2}}P+\tilde{r}(u;y_{1},y_{2})$
where $\tilde{r}(u;y_{1},y_{2})$ does not have a pole along $y_{1}=y_{2}$.
Let us write $r(u;y_{ij})$ and $\tilde{r}(u;y_{ij})$ as short-hand
for $r(u;y_{i},y_{j})$ and $\tilde{r}(u;y_{i},y_{j})$ respectively.
Then starting from (\ref{eq:Dual AYBE in 3 variables}) we derive
that for $v=-u+h$:\begin{equation}
\begin{array}{c}
{\displaystyle \frac{{}}{{}}}r^{13}(u;y_{13})\, r^{23}(-u+h;y_{23})=r^{23}(h;y_{23})\, r^{12}(u;y_{12})-r^{12}(u-h;y_{12})\, r^{13}(h;y_{13})=\\
{\displaystyle \frac{{}}{{}}}=\left[r^{23}(h;y_{23})\, r^{12}(u;y_{12})-r^{12}(u;y_{12})\, r^{13}(h;y_{13})\right]+\\
{\displaystyle \frac{{}}{{}}}+\left[r^{12}(u;y_{12})-r^{12}(u-h;y_{12})\right]\, r^{13}(h;y_{13}).\end{array}\label{eq:gen_massey_on_cycles_5}\end{equation}
 Let us rewrite the expression in the first bracket on the right-most
side as \[
\left(r^{23}(h;y_{23})\,\frac{c}{y_{1}-y_{2}}P^{12}-\frac{c}{y_{1}-y_{2}}P^{12}\, r^{13}(h;y_{13})\right)+\]
 \[
+r^{23}(h;y_{23})\,\tilde{r}^{12}(u;y_{12})-\tilde{r}^{12}(u;y_{12})\, r^{13}(h;y_{13}).\]
 Using fact \ref{fac:P is the center}, we know that $P^{12}\, r^{13}(h;y_{13})=r^{23}(h;y_{13})\, P^{12}$,
hence the right-most side of (\ref{eq:gen_massey_on_cycles_5}) equals\[
\frac{r^{23}(h;y_{23})-r^{23}(h;y_{13})}{y_{1}-y_{2}}\, cP^{12}+r^{23}(h;y_{23})\,\tilde{r}^{12}(u;y_{12})-\]
 \[
-\tilde{r}^{12}(u;y_{12})\, r^{13}(h;y_{13})+\left[\tilde{r}^{12}(u;y_{12})-\tilde{r}^{12}(u-h;y_{12})\right]\, r^{13}(h;y_{13}).\]
 Passing to the limit $y_{2}\rightarrow y_{1}$, we see that\begin{equation}
\begin{array}{c}
{\displaystyle \frac{{}}{{}}}r^{13}(u;y_{13})\, r^{23}(-u+h;y_{13})=-{\displaystyle \frac{\partial r^{23}}{\partial y_{1}}}\left(h;y_{13}\right)\, cP^{12}+r^{23}(h;y_{13})\,\tilde{r}^{12}(u;y_{11})-\\
{\displaystyle \frac{{}}{{}}}-\tilde{r}^{12}(u;y_{11})\, r^{13}(h;y_{13})+\left[\tilde{r}^{12}(u;y_{11})-\tilde{r}^{12}(u-h;y_{11})\right]\, r^{13}(h;y_{13}).\end{array}\label{eq:gen_massey_on_cycles_6}\end{equation}
 We want to apply the operator $\mu\otimes\mbox{id}:A\otimes A\otimes A\rightarrow A\otimes A$
to this equation, where $\mu$ is the product in $A$. Observe that
\[
\left(\mu\otimes\mbox{id}\right)\left(a^{13}b^{23}\right)=ab,\,\left(\mu\otimes\mbox{id}\right)(a^{23}b^{12}-b^{12}a^{13})=0\]
 where $a,b\in A\otimes A$ and the notation is best explained by
the example $a^{13}=a_{1}\otimes\mathds{1}\otimes a_{2}$ for $a=a_{1}\otimes a_{2}$.
Moreover, using that $\sum_{i,j}e_{ij}ae_{ji}=\mbox{tr}(a)\mathds{1}$
for any $a\in A$ clearly, we derive that for $\mbox{tr}_{1}=\mbox{tr}\otimes\mbox{id}:A\otimes A\rightarrow A$
we have \[
\left(\mu\otimes\mbox{id}\right)(a^{23}P^{12})=\mathds{1}\otimes\mbox{tr}_{1}(a).\]
 Hence applying $\mu\otimes\mbox{id}$ to (\ref{eq:gen_massey_on_cycles_6})
yields \[
r(u;y_{13})\, r(-u+h;y_{13})=-c\cdot\,\mathds{1}\otimes\mbox{tr}_{1}\left(\frac{\partial r}{\partial y_{1}}\left(h;y_{13}\right)\right)+\]
 \[
+\left(\mu\otimes\mbox{id}\right)\left(\left[\tilde{r}^{12}(u;y_{11})-\tilde{r}^{12}(u-h;y_{11})\right]\, r^{13}(h;y_{13})\right).\]
 Now, take the limit $h\rightarrow0$. The left-hand side of yields
$s(u;y_{1},y_{3})$. As for the right-hand side, we invoke our assumption
on the existence of a certain Laurent expansion (\ref{eq:Laurent assumption})
to derive that \[
\lim_{h\rightarrow0}\frac{\partial r}{\partial y_{1}}\left(h;y_{13}\right)=\frac{\partial r_{0}}{\partial y_{1}}\left(y_{1},y_{3}\right).\]
 Moreover\[
\lim_{h\rightarrow0}\left(\left[\tilde{r}^{12}(u;y_{11})-\tilde{r}^{12}(u-h;y_{11})\right]\, r^{13}(h;y_{13})\right)=\frac{\partial\tilde{r}^{12}}{\partial u}\left(u;y_{11}\right)\left(\lim_{h\rightarrow0}r^{13}(h;y_{13})\cdot h\right).\]
 Again using (\ref{eq:Laurent assumption}), we see that the second
factor of this last term is simply $\mathds{1}\otimes\mathds{1}\otimes\mathds{1}$.
Putting all this together, we end up with\[
s(u;y_{1},y_{3})=-c\cdot\,\mathds{1}\otimes\mbox{tr}_{1}\left(\frac{\partial r_{0}}{\partial y_{1}}\left(y_{1},y_{3}\right)\right)+\mu\left(\frac{\partial\tilde{r}}{\partial u}\left(u;y_{1},y_{1}\right)\right)\otimes\mathds{1}.\]
 Hence we may write $s(u;y_{1},y_{2})=\mathds{1}\otimes\beta(y_{1},y_{2})+\gamma(u,y_{1})\otimes\mathds{1}$.
Note that $\beta(y_{1},y_{2})=\mbox{pr}\left(\beta(y_{1},y_{2})\right)+\frac{\mbox{tr}\,\left(\beta(y_{1},y_{2})\right)}{n}\,\mathds{1}$.
Using the same trick for $\gamma(u,y_{1})$, we may actually write\[
s(u;y_{1},y_{2})=a(u,y_{1})\otimes\mathds{1}+\mathds{1}\otimes b(y_{1},y_{2})+\left(f(u,y_{1})+g(y_{1},y_{2})\right)\mathds{1}\otimes\mathds{1}\]
 where now both \[
a(u,y_{1})=\mbox{pr}\,\mu\left(\frac{\partial\tilde{r}}{\partial u}\left(u;y_{1},y_{1}\right)\right),\, b(y_{1},y_{2})=-c\cdot\mbox{pr }\mbox{tr}_{1}\left(\frac{\partial r_{0}}{\partial y_{1}}\left(y_{1},y_{2}\right)\right)\]
 map to $\mathfrak{sl}_{n}(\mathbb{C})$. Note that unitarity of $r(u;y_{1},y_{2})$
implies that $s^{21}(-u;y_{2},y_{1})=s^{12}(u;y_{1},y_{2})$. Applying
$\mbox{pr}\otimes\mathds{1}$ to this equation yields $a(u,y_{1})=b(y_{2},y_{1})$.
It follows that both $a$ and $b$ depend on the second variable only
and actually coincide, hence\[
s(u;y_{1},y_{2})=a(y_{1})\otimes\mathds{1}+\mathds{1}\otimes a(y_{2})+\left(f(u,y_{1})+g(y_{1},y_{2})\right)\mathds{1}\otimes\mathds{1}.\]
 In order to show the statement concerning the form of $s$, we have
to prove that $a(y_{1})$ is constant. To this end, we substitute
the form of $s$ just calculated into the second equation of the equality
stated in lemma \ref{lem:QYBE in terms of s*r-r*s}. We derive that\begin{equation}
\begin{array}{c}
{\displaystyle \frac{{}}{{}}}\left[a^{1}(y_{1})+a^{3}(y_{3}),r^{13}(u_{13};y_{1},y_{3})\right]=r^{13}(u_{13};y_{1},y_{3})\cdot\\
{\displaystyle \frac{{}}{{}}}\cdot\left(f(u_{32},y_{1})+f(u_{21},y_{2})-f(u_{12},y_{1})-f(u_{23},y_{2})\right).\end{array}\label{eq:gen_massey_on_cycles_7}\end{equation}
 Let us focus on the left-hand side. This equals\[
\left[a^{1}(y_{1})+a^{3}(y_{3}),\frac{cP^{13}}{y_{1}-y_{3}}\right]+\left[a^{1}(y_{1})+a^{3}(y_{3}),\tilde{r}^{13}(u_{13};y_{1},y_{3})\right].\]
 By fact \ref{fac:P is the center}, we may rewrite the first summand
as\[
\frac{a^{1}(y_{1})-a^{1}(y_{3})}{y_{1}-y_{3}}cP^{13}+\frac{a^{3}(y_{3})-a^{3}(y_{1})}{y_{1}-y_{3}}cP^{13}\]
 thus the limit $y_{3}\rightarrow y_{1}$ of the left-hand side of
(\ref{eq:gen_massey_on_cycles_7}) is given by\[
\frac{d}{dy_{1}}\left(a^{1}(y_{1})-a^{3}(y_{1})\right)cP^{13}+\left[a^{1}(y_{1})+a^{3}(y_{1}),\tilde{r}^{13}(u_{13};y_{1},y_{1})\right].\]
 In particular, the limit $y_{3}\rightarrow y_{1}$ of the right-hand
side of (\ref{eq:gen_massey_on_cycles_7}) exists as well. But $r^{13}(u_{13};y_{1},y_{3})$
has a pole along $y_{1}=y_{3}$ and the other factor on the right-hand
side is independent of $y_{3}$. Hence we conclude that \[
f(u_{32},y_{1})+f(u_{21},y_{2})-f(u_{12},y_{1})-f(u_{23},y_{2})=0.\]
 In particular, the left-hand side of (\ref{eq:gen_massey_on_cycles_7})
equals zero. Focusing on the polar part yields \[
\left[a^{1}(y_{1})+a^{3}(y_{3}),\frac{cP^{13}}{y_{1}-y_{3}}\right]=0\]
 which, by the above, implies\[
\frac{d}{dy_{1}}\left(a^{1}(y_{1})-a^{3}(y_{1})\right)cP^{13}=0.\]
 But then $\frac{da}{dy}(y)$ must be zero, so that $a(y)=a\in\mathfrak{sl}_{n}(\mathbb{C})$
is constant. Therefore $s(u;y_{1},y_{2})=a\otimes\mathds{1}+\mathds{1}\otimes a+\left(f(u,y_{1})+g(y_{1},y_{2})\right)\mathds{1}\otimes\mathds{1}$.
By (\ref{eq:gen_massey_on_cycles_7}) we also see that $a$ is an
infinitesimal symmetry of $r(u;y_{1},y_{2})$.\\
  Finally, we want to prove the statement concerning $r_{0}(y_{1},y_{2})$.
Clearly, we may write \[
r_{0}(y_{1},y_{2})=\overline{r}_{0}(y_{1},y_{2})+\alpha(y_{2},y_{1})\otimes\mathds{1}-\mathds{1}\otimes\alpha(y_{1},y_{2})+h(y_{1},y_{2})\mathds{1}\otimes\mathds{1}\]
 with $\alpha$ mapping to $\mathfrak{sl}_{n}(\mathbb{C})$. Note
that in the discussion above we derived that \[
a=b(y_{1},y_{2})=-c\cdot\mbox{pr }\mbox{tr}_{1}\left(\frac{\partial r_{0}}{\partial y_{1}}\left(y_{1},y_{2}\right)\right).\]
 Since both $\overline{r}_{0}$ and $\alpha$ map to $\mathfrak{sl}_{n}(\mathbb{C})$,
so will their partial derivatives. This implies\[
a=-cn\cdot\mbox{pr}\left(-\left(\frac{\partial}{\partial y_{1}}\right)\alpha(y_{1},y_{2})+\frac{\partial h}{\partial y_{1}}(y_{1},y_{2})\cdot\mathds{1}\right).\]
 Hence $\frac{\partial}{\partial y_{1}}\alpha(y_{1},y_{2})=\frac{a}{cn}$,
which gives the formula for $\alpha$. In particular, $\alpha(y_{1},y_{2})$
does not depend on the second argument. This completes the proof of
the formula for $r_{0}$. 
\end{proof}

\noindent Before we finally prove theorem \ref{thm: CYBE -> QYBE -> CYBE},
we need to state one more easy fact:

\begin{fact} \emph{
\textup{\label{fact: pr a otimes b phi lemma  and other}} \textup{a)
For any}\emph{ }\textup{\emph{$x,y,\phi\in A$,}}\emph{ }\textup{$i\in\{1,2\}$
and $\phi^{1}=\phi\otimes\mathds{1}$ respectively $\phi^{2}=\mathds{1}\otimes\phi$,
we have\[
(\mbox{pr}\otimes\mbox{pr})\left[x\otimes y,\phi^{i}\right]=\left[(\mbox{pr}\otimes\mbox{pr})\left(x\otimes y\right),\phi^{i}\right].\]
}
}

\textup{b)} \emph{Let $r$ be a non-degenerate solution of (\ref{eq:AYBE in 3 variables})
and $\left[r,1\otimes a\right]=0$ for some $a\in\mathfrak{sl}_{n}(\mathbb{C})$.
Then $a=0$.} 
\end{fact}
\begin{proof}
a) is straightforward. As to b), write $r=\sum_{i\in I}r'_{i}\otimes r''_{i}$
for some index set $I$ and let $\varphi:A\otimes A\rightarrow\mbox{End}(A)$
denote the isomorphism given by $X\otimes Y\mapsto\left(Z\mapsto X\,\mbox{tr}\left(YZ\right)\right)$.
Then $$0=\varphi\left(\left[r,1\otimes a\right]\right)(b)=\sum_{i\in I}r'_{i}\,\mbox{tr}\left(\left[r''_{i},a\right]b\right)=\sum_{i\in I}r'_{i}\,\mbox{tr}\left(r''_{i}\left[a,b\right]\right)=\varphi(r)\left(\left[a,b\right]\right)$$
for all $b\in A$. Now $r$ is non-degenerate, hence $\varphi(r)$
is an isomorphism. This yields $ $$\left[a,b\right]=0$ for all $b\in A$,
especially all $b\in\mathfrak{sl}_{n}(\mathbb{C})$. But the Lie bracket
is non-degenerate on $\mathfrak{sl}_{n}(\mathbb{C})$, hence $a=0$. 
\end{proof}
\noindent  \emph{Proof of Theorem \label{proof: of main}\ref{thm: CYBE -> QYBE -> CYBE}.}

\noindent (1) By lemma \ref{lem: AYBE unitary ->CYBE unitary } $\bar{r}_{0}$
is a unitary solution of the CYBE (\ref{eq:CYBE}). The rest is immediate
by lemma \ref{cor:simple poles exist} and the fact that $\left(\mbox{pr}\otimes\mbox{pr}\right)\left(P\right)$
is the Casimir element of $\mathfrak{sl}_{n}(\mathbb{C})\otimes\mathfrak{sl}_{n}(\mathbb{C})$
with respect to the trace form $(x,y)\mapsto\mbox{tr}(x\cdot y)$.\\

\noindent  (2) Setting $u_{1}=u$, $u_{2}=0$ and $u_{3}=-u$ in lemma \ref{lem:QYBE in terms of s*r-r*s},
we derive that $r(u;y_{1},y_{2})$ satisfies the QYBE (\ref{eq:QYBE})
for $u$ fixed if and only if\[
s^{23}(u;y_{2},y_{3})\, r^{13}(2u;y_{1},y_{3})=r^{13}(2u;y_{1},y_{3})\, s^{23}(-u;y_{2},y_{3}).\]
 Applying lemma \ref{lem:form of s} this is equivalent to\begin{equation}
\left[r(u;y_{1},y_{2}),1\otimes a\right]=0\label{eq: final prove, a is inf. sym}\end{equation}
which, by fact \ref{fact: pr a otimes b phi lemma  and other} b)
is equivalent to $a=0$. By lemma \ref{lem:form of s} this last condition
holds if and only if either of conditions (b) or (c) of the theorem
are satisfied. It remains to show equivalence with condition (d).
To this end, recall (\ref{eq:r1 determined by r0}). Denote the right-hand
side of this equation by $AYBE[r_{0}](y_{1},y_{2},y_{3})$. Then (d)
simply reads $(\mbox{pr}\otimes\mbox{pr}\otimes\mbox{pr})\, AYBE[\overline{r}_{0}](y_{1},y_{2},y_{3})=0$.
To show the equivalence of this with $a=0$, express $r_{0}$ in terms
of $\overline{r}_{0}$ as stated in lemma \ref{lem:form of s}. Then
\[
-cn\cdot\left(\mbox{pr}\otimes\mbox{pr}\otimes\mbox{pr}\right)\left(AYBE\left[r_{0}\right](y_{1},y_{2},y_{3})-AYBE\left[\overline{r}_{0}\right](y_{1},y_{2},y_{3})\right)=\]
 \[
=\left(y_{1}-y_{3}\right)\overline{r}_{0}^{13}(y_{1},y_{3})\, a^{2}+\left(y_{1}-y_{2}\right)\overline{r}_{0}^{12}(y_{1},y_{2})\, a^{3}+\left(y_{3}-y_{2}\right)\overline{r}_{0}^{23}(y_{2},y_{3})\, a^{1}.\]
 It is immediate by (\ref{eq:r1 determined by r0}) that $\left(\mbox{pr}\otimes\mbox{pr}\otimes\mbox{pr}\right)AYBE\left[r_{0}\right](y_{1},y_{2},y_{3})=0$.
Hence if $a$ is zero, this implies $(\mbox{pr}\otimes\mbox{pr}\otimes\mbox{pr})\, AYBE[\overline{r}_{0}]=0$.
On the other hand, assuming $(\mbox{pr}\otimes\mbox{pr}\otimes\mbox{pr})\, AYBE[\overline{r}_{0}]=0$
we deduce \[
\left(y_{1}-y_{3}\right)\overline{r}_{0}^{13}(y_{1},y_{3})\, a^{2}+\left(y_{1}-y_{2}\right)\overline{r}_{0}^{12}(y_{1},y_{2})\, a^{3}+\left(y_{3}-y_{2}\right)\overline{r}_{0}^{23}(y_{2},y_{3})\, a^{1}=0.\]
 We will show that this implies $a=0$. Indeed, by lemma \ref{cor:simple poles exist}
we know that $r_{0}(y_{1},y_{2})=\frac{cP}{y_{1}-y_{2}}+\tilde{r}_{0}(y_{1},y_{2})$
with $\tilde{r}_{0}$ being defined along $y_{1}=y_{2}$ and similarly
for $\overline{r}_{0}$. Hence passing to the limit $y_{1},y_{2},y_{3}\rightarrow y$
yields \[
\left(\mbox{pr}\otimes\mbox{pr}\otimes\mbox{pr}\right)\left[P^{13}a^{2}+P^{12}a^{3}+P^{23}a^{1}\right]=0.\]
 Let us write $a=\sum a_{ij}e_{ij}$. Looking at the coefficient of
$e_{ij}\otimes e_{ji}\otimes e_{ij}$ in the above equation for $i\neq j$,
we derive $a_{ij}=0$. But then projecting the above equation to $e_{12}\otimes e_{21}\otimes\mathfrak{sl}_{n}(\mathbb{C})$,
we may deduce that $a=0$.\\

\noindent  (3) As we just saw, the conditions of (2) are satisfied if $a=0$.
But $a$ is an infinitesimal symmetry of $r$ by lemma \ref{lem:form of s},
hence one of $r_{0}$. Invoking fact \ref{fact: pr a otimes b phi lemma  and other}
a), we deduce that $a$ is an infinitesimal symmetry of $\overline{r}_{0}$
and so $a=0$. $\hfill\Box$

\section{Uniqueness of lifts from CYBE to AYBE}

In this section we will prove part ii) of theorem \ref{thm: generalizing Polishchuk}:

\begin{thm}
\emph{\label{thm:recovering AYBE sol. from their CYBE}\cite[Theorem 6]{Polischuck2002}}
Let $r(u;y_{1},y_{2})$ and $s(u;y_{1},y_{2})$ be a unitary solutions
of (\ref{eq:AYBE in 3 variables}) of the form (\ref{eq:Laurent assumption}).
Assume that the corresponding solution \emph{$\overline{r}_{0}(y_{1},y_{2})=(\mbox{pr}\otimes\mbox{pr})\left(r_{0}(y_{1},y_{2})\right)$}
of the CYBE (\ref{eq:CYBE}) is non-degenerate, has no infinitesimal
symmetries and that $\overline{s}_{0}(y_{1}y_{2})=\overline{r}_{0}(y_{1},y_{2})$.
Then there exists a meromorphic function $g:\mathbb{C}\rightarrow\mathbb{C}$
such that $s(u;y_{1},y_{2})=\exp\bigl(u\bigl(g(y_{2})-g(y_{1})\bigr)\bigr)r(u;y_{1},y_{2})$. 
\end{thm}
\begin{proof}
First, we show that $r$ is uniquely determined by $r_{0}$. By lemma
\ref{lem: r uniquely determined by r_0,r_1} $r$ is uniquely determined
by $r_{0}$ and $r_{1}$ and moreover $r_{1}$ is a solution of a
certain equation in $r_{0}$ which is given by (\ref{eq:r1 determined by r0}).
If $r'_{1}\neq r_{1}$ was a solution of (\ref{eq:r1 determined by r0})
with the same properties as $r_{1}$, then taking the difference we
would obtain a meromorphic function $\alpha:\left(\mathbb{C}^{2},0\right)\rightarrow A\otimes A$
with $\alpha^{21}(y_{2},y_{1})=\alpha(y_{1},y_{2})$ and \begin{equation}
\alpha^{12}(y_{1},y_{2})+\alpha^{13}(y_{1},y_{3})+\alpha^{23}(y_{2},y_{3})=0.\label{eq:recovering_AYBE_from_CYBE_3}\end{equation}
 Using lemma \ref{lem:poles are simple}, we also know that the residue
of $r(u;y_{1},y_{2})$ near $y_{1}=y_{2}$ is independent of $u$.
Comparing this to the Laurent expansion (\ref{eq:Laurent assumption}),
we derive that $r_{1}(y_{1},y_{2})$ has no poles along $y_{1}=y_{2}$,
hence the same is true for $\alpha(y_{1},y_{2})$. To prove that $r_{1}$
is determined by $r_{0}$, we need only show that $\alpha$ is already
zero. Choosing $y_{3}=y_{2}$ and then applying $\mbox{pr}\otimes\mbox{id}\otimes\mbox{id}$
to (\ref{eq:recovering_AYBE_from_CYBE_3}) we derive that $(\mbox{pr}\otimes\mbox{id})\left(\alpha(y_{1},y_{2})\right)=0$.
Similarly, $(\mbox{id}\otimes\mbox{pr})\left(\alpha(y_{1},y_{2})\right)=0$,
hence $\alpha(y_{1},y_{2})=f(y_{1},y_{2})\,\mathds{1}\otimes\mathds{1}$
where $f$ is a meromorphic function such that $f(y_{1},y_{2})+f(y_{1},y_{3})+f(y_{2},y_{3})=0$.
Since $r_{1}(y_{1},y_{2})$ has no pole along $y_{1}=y_{2}$ by lemma
\ref{lem:poles are simple}, $\alpha(y_{1},y_{1})$ exists. We may
deduce that $2f(y_{1},y_{2})=-f(y_{2},y_{2})$, so $f$ depends only
on the second variable. But then choosing $y_{2}=y_{1}=y_{3}$ we
read $3f(y_{1},y_{1})=0$, thus $f=0$.We have proved that $r$ is
uniquely determined by $r_{0}$.\\

\noindent It remains to prove that, provided $\overline{r}_{0}$ has
no infinitesimal symmetries, $r$ can be uniquely recovered from $\overline{r}_{0}$
up to the factor $\exp\left(u\left(g(y_{2})-g(y_{1})\right)\right)$
for some meromorphic function $g:\mathbb{C}\rightarrow\mathbb{C}$.
Note that this is equivalent to showing that $r_{0}(y_{1},y_{2})$
is uniquely determined by $\overline{r}_{0}(y_{1},y_{2})=(\mbox{pr}\otimes\mbox{pr})\left(r_{0}(y_{1},y_{2})\right)$
up to a summand of the form $\left(g(y_{2})-g(y_{1})\right)\,\mathds{1}\otimes\mathds{1}$.
By assumption $\left(s_{0}(y_{1},y_{2}),s_{1}(y_{1},y_{2})\right)$
is another tuple satisfying (\ref{eq:r1 determined by r0}) such that\[
s_{0}^{21}(y_{2},y_{1})=-s_{0}(y_{1},y_{2}),\, s_{1}^{21}(y_{2},y_{1})=s_{1}(y_{1},y_{2}).\]
 We claim that \[
s_{0}(y_{1},y_{2})=r_{0}(y_{1},y_{2})+\left(g(y_{2})-g(y_{1})\right)\,\mathds{1}\otimes\mathds{1}.\]
 Since $\overline{s}_{0}(y_{1},y_{2})=\overline{r}_{0}(y_{1},y_{2})$,
we may write \[
s_{0}(y_{1},y_{2})=r_{0}(y_{1},y_{2})+\phi^{1}(y_{1},y_{2})-\phi^{2}(y_{2},y_{1})+\psi(y_{1},y_{2})\,\mathds{1}\otimes\mathds{1}\]
 for a $\mathfrak{sl}_{n}(\mathbb{C})$ valued function $\phi$ and
a scalar function $\psi$. Denoting the left-hand side of (\ref{eq:r1 determined by r0})
by $LHS(r)$, we have \[
0=\left(\mbox{pr}\otimes\mbox{pr}\otimes\mbox{pr}\right)(LHS(s)-LHS(r))=\overline{r}_{0}^{12}(y_{1},y_{2})\,\left[\phi^{3}(y_{3},y_{2})-\phi^{3}(y_{3},y_{1})\right]+\]
 \[
+\overline{r}_{0}^{23}(y_{2},y_{3})\,\left[\phi^{1}(y_{1},y_{3})\right.\left.-\phi^{1}(y_{1},y_{2})\right]+\overline{r}_{0}^{13}(y_{1},y_{3})\,\left[\phi^{2}(y_{2},y_{3})-\phi^{2}(y_{2},y_{1})\right].\]
 If the function $\phi$ is not constant then contracting this equation
with a generic functional in the third component we derive that $\overline{r}_{0}$
is a sum of two decomposable tensors, that is $\bar{r}_{0}=a_{1}\otimes b_{1}+a_{2}\otimes b_{2}$
where all terms depend on $y_{1},y_{2}$. But $\bar{r}_{0}$ is non-degenerate
by assumption, so $\mbox{span}_{\mathbb{C}}\left(\left\{ a_{1},a_{2}\right\} \right)\cong\mathfrak{sl}_{n}(\mathbb{C})$,
which is impossible for any $n\geq2$. Thus $\phi\in\mathfrak{sl}_{n}(\mathbb{C})$
is constant. Applying $(\mbox{pr}\otimes\mbox{pr}\otimes\mbox{id})$
to $LHS(s)-LHS(r)$ yields\begin{equation}
\begin{array}{c}
{\displaystyle \frac{{}}{{}}}(\mbox{pr}\otimes\mbox{pr}\otimes\mbox{id})\left(s_{1}^{12}(y_{1},y_{2})-r_{1}^{12}(y_{1},y_{2})\right)=(\mbox{pr}\otimes\mbox{pr}\otimes\mbox{id})\left(r_{0}^{12}(y_{1},y_{2})\phi^{1}-\right.\\
\left.{\displaystyle \frac{{}}{{}}}-\phi^{2}r_{0}^{12}(y_{1},y_{2})\right)-\phi^{1}\phi^{2}+\left(\psi(y_{1},y_{3})-\psi(y_{2},y_{3})\right)\,\overline{r}^{12}(y_{1},y_{2}).\end{array}\label{eq:recovering_AYBE_from_CYBE_4}\end{equation}
 This implies that $\psi(y_{1},y_{3})-\psi(y_{2},y_{3})$ is actually
independent of $y_{3}$, hence equal to some function $\beta(y_{1},y_{2})$.
Also, we know by unitarity of $r$ that $\psi(y_{1},y_{2})=-\psi(y_{2},y_{1})$,
thus $\beta(y_{1},y_{2})=\psi(y_{1},y_{3})+\psi(y_{3},y_{2})$. It
follows from lemma \ref{lem:poles are simple} that $r_{0}$ and $s_{0}$
have the same pole along $y_{1}=y_{2}$, hence $\psi(y_{1},y_{1})$
exists and we may deduce that $\beta(y_{1},y_{2})=\psi(y_{1},y_{1})+\psi(y_{1},y_{2})=\psi(y_{1},y_{2})$.
Thus the definition of $\beta$ reads $\psi(y_{1},y_{2})=\psi(y_{1},y_{3})-\psi(y_{2},y_{3})$.
Therefore, defining $g(y)=\psi(y,a)$ for some fixed $a\in\mathbb{C}$,
we have $\psi(y_{1},y_{2})=g(y_{1})-g(y_{2})$. Altogether \begin{equation}
s_{0}(y_{1},y_{2})=r_{0}(y_{1},y_{2})+\phi^{1}-\phi^{2}+\left(g(y_{1})-g(y_{2})\right)\,\mathds{1}\otimes\mathds{1}\label{eq: s0 in terms of r0}\end{equation}
 Since $s_{0}$ and $r_{0}$ are both meromorphic, so is $\psi$ and
thus also $g$.

Next, we replace $r(u;y_{1},y_{2})$ by $\exp\left(u\left(g(y_{2})-g(y_{1})\right)\right)r(u;y_{1},y_{2})$
and hence may assume that $g=0$ in the above formula for $s_{0}$.
Thus (\ref{eq:recovering_AYBE_from_CYBE_4}) yields\[
(\mbox{pr}\otimes\mbox{pr})\left(s_{1}(y_{1},y_{2})-r_{1}(y_{1},y_{2})\right)=(\mbox{pr}\otimes\mbox{pr})\left(r_{0}(y_{1},y_{2})\phi^{1}-\phi^{2}r_{0}(y_{1},y_{2})\right)-\phi^{1}\phi^{2}.\]
 We exchange the first two components, make the substitutions $y_{1}\leftrightarrow y_{2},y_{2}\leftrightarrow y_{1}$
and use unitarity of $r$ for both sides of the resulting equation.
Comparing the result with the above equation, we derive \[
(\mbox{pr}\otimes\mbox{pr})\left(r_{0}(y_{1},y_{2})\phi^{1}-\phi^{2}r_{0}(y_{1},y_{2})\right)=(\mbox{pr}\otimes\mbox{pr})\left(-r_{0}(y_{1},y_{2})\phi^{2}+\phi^{1}r_{0}^{}(y_{1},y_{2})\right).\]
 By fact \ref{fact: pr a otimes b phi lemma  and other} a) we deduce
that $\left[\overline{r}_{0}(y_{1},y_{1}),\phi^{1}+\phi^{2}\right]=0$.
But then $\phi$ is an infinitesimal symmetry of $\overline{r}_{0}$,
so $\phi=0$. Thus $s_{0}=r_{0}$. 
\end{proof}
\begin{rem}
By theorem \ref{thm: CYBE -> QYBE -> CYBE} (1), the assumption of
theorem \ref{thm:recovering AYBE sol. from their CYBE} on the non-degeneracy
of $\bar{r}_{0}$ is automatically satisfied if $r$ itself is non-degenerate.
Moreover, we deduce from the proof of theorem \ref{thm:recovering AYBE sol. from their CYBE}
that in that case $r$ is already uniquely determined by $r_{0}$.
\end{rem}
\begin{cor}
In the notations of theorem \ref{thm:recovering AYBE sol. from their CYBE},
assume that $r_{0}(y_{1},y_{2})$ and $s_{0}(y_{1},y_{2})$ have the
same poles on $\bigl(\mathbb{C}^{2}\setminus V\bigl((y_{1}-y_{2})\bigr)\bigr)$.
Then $g$ is a holomorphic function. Thus, $r$ and $s$ are gauge
equivalent. 
\end{cor}
\begin{proof}
It follows from the assumption and lemma \ref{lem:poles are simple}
that the poles of $r_{0}$ and $s_{0}$ coincide. By (\ref{eq: s0 in terms of r0})
this implies that $g$ is holomorphic. The remaining statement follows
from the discussion in example \ref{exa: example for gauge equivalence}. 
\end{proof}
\noindent Combining the above corollary with theorem \ref{thm: r(n,d) is unitary non-deg AYBE solution}
yields the final result of this section, finishing the proof of theorem
\ref{thm: generalizing Polishchuk} ii).

\begin{cor}
\label{cor: if r and s are both constructed geometrically}In the
notations of theorem \ref{thm:recovering AYBE sol. from their CYBE},
assume that both $r$ and $s$ can be obtained by the procedure described
in section \ref{sec:The-procedure}. Then $r$ and $s$ are gauge
equivalent. 
\end{cor}

\end{document}